\documentclass[onecolumn,aps,pra,nofootinbib,superscriptaddress,tightenlines,notitlepage,12pt]{revtex4-1}

\usepackage{caption}

\makeatletter
\renewcommand{\p@subsection}{}
\renewcommand{\p@subsubsection}{}
\makeatother

\makeatletter
\def\l@subsubsection#1#2{}
\makeatother

\usepackage{amsmath,amsfonts,amsthm}
\usepackage{mathrsfs}
\usepackage{amssymb}
\usepackage{subcaption}
\usepackage{caption}
\usepackage{enumerate}
\usepackage{float}
\usepackage{mathtools}
\usepackage[left=0.75in,right=0.75in,top=0.75in,bottom=0.75in]{geometry}
\usepackage{color}
\usepackage{graphicx}
\usepackage[breaklinks,colorlinks = true,linkcolor = blue,urlcolor  = blue,citecolor = red,anchorcolor = blue]{hyperref}
\usepackage[matrix,frame,arrow]{xypic}
\usepackage[braket]{qcircuit}
\usepackage{slashbox}
\usepackage{comment}
\usepackage{tikz}

\usepackage{verbatim}
\usepackage[framemethod=tikz]{mdframed}

\usepackage{algorithm}
\usepackage{algpseudocode}

\mathtoolsset{showonlyrefs}

\theoremstyle{plain}
\newtheorem*{result*}{Main Result}

\newtheorem{lemma}{Lemma}

\newtheorem{theorem}[lemma]{Theorem}
\newtheorem{remark}[lemma]{Remark}
\newtheorem{fact}[lemma]{Fact}
\newtheorem{claim}[lemma]{Claim}

\newtheorem{assume}[lemma]{Assumption}

\theoremstyle{definition}

\newtheorem{definition}[lemma]{Definition}

\newcommand{\eps}{\varepsilon}

\newcommand{\verifier}[0]{\mathcal{V}}
\newcommand{\prover}[0]{\mathcal{P}}

\newcommand{\be}{\begin{equation}}
\newcommand{\ee}{\end{equation}}

\newcommand{\ihpc}{Institute of High Performance Computing (IHPC), Agency for Science,
Technology and Research (A*STAR), 1 Fusionopolis Way,
\# 16-16 Connexis, Singapore 138632, Republic of Singapore}
\newcommand{\affilR}{Centre for Quantum Technologies and Department of Computer Science, National University of Singapore and
MajuLab, UMI 3654, Singapore}

\newcommand{\suppress}[1]{}

\mdfdefinestyle{mystyle}{leftmargin=0cm, innerleftmargin=-0.3cm}

\begin{document}

\title{\Large On the power of geometrically-local classical and quantum circuits}

\author{Kishor Bharti}
\email{kishor.bharti1@gmail.com}
\affiliation{\ihpc}

\author{Rahul Jain}
\email{rahul@comp.nus.edu.sg}
\affiliation{\affilR}

\begin{abstract}
We show a relation, based on parallel repetition of the Magic Square game, that can be solved, with probability exponentially close to $1$ (worst-case input), by $1D$ (uniform) depth $2$, geometrically-local, noisy (noise below a threshold), fan-in $4$, quantum circuits. We show that the same relation cannot be solved, with an exponentially small success probability (averaged over inputs drawn uniformly), by $1D$ (non-uniform) geometrically-local,  sub-linear depth, classical circuits consisting of fan-in $2$ NAND gates. Quantum and classical circuits are allowed to use input-independent (geometrically-non-local) resource states, that is entanglement and randomness respectively. To the best of our knowledge, previous best (analogous) depth separation for a task between quantum and classical  
circuits was constant v/s sub-logarithmic, although for general (geometrically non-local) circuits.

Our hardness result for classical circuits is based on a direct product theorem about classical communication protocols from Jain and Kundu~\cite{jain2022direct}. 

As an application, we propose a protocol that can potentially demonstrate verifiable quantum advantage in the NISQ era. We also provide generalizations of our result for higher dimensional circuits as well as a wider class of Bell games.   
\end{abstract}
\maketitle

\section{Introduction}
Presently, quantum computing systems offer a few hundred imperfect qubits with limited connections, constituting what is known as Noisy Intermediate-Scale Quantum (NISQ)~\cite{preskill2018quantum,bharti2021noisy,cerezo2021variational,deutsch2020harnessing} devices. During the NISQ era, a key goal has been to demonstrate quantum computational advantage, wherein a quantum computing device achieves experimental success in solving a classically hard problem. This not only represents a remarkable experimental milestone but also potentailly challenges the extended Church-Turing thesis~\cite{arora2009computational}.

Efforts have been directed towards showing separation between efficient quantum computation from efficient classical computation based on complexity theoretic conjectures~\cite{aaronson2011computational,bremner2017achieving,farhi2016quantum,bremner2016average}. A complementary direction has been to focus on quantum computational advantages in a regime where classical simulations are still efficient and unconditional results are feasible~\cite{bravyi2018quantum,coudron2021trading,gall2018average,watts2019expo,bravyi2020quantum}. 

\subsection{Results}
We discuss a relation problem based on parallel repetition of the Magic Square game and use the same in conjunction with geometric locality (see subsection~\ref{subsec: geometric_locality}) and initial resource state (see subsection~\ref{subsec: initial_resource} )  to arrive at the following result.

\begin{result*} [informal; see Theorem~\ref{thm: main}] There exists a relation, based on parallel repetition of the Magic Square game, that can be solved, with probability exponentially close to $1$ (worst-case input), by $1D$ (uniform) depth $2$, geometrically-local, noisy (noise below a threshold), fan-in $4$, entanglement assisted quantum circuits.

The same relation cannot be solved, with an exponentially small success probability (averaged over the inputs drawn uniformly), by $1D$ (non-uniform) geometrically-local, sub-linear depth, randomness-assisted, classical circuits consisting of fan-in $2$ NAND gates (note that NAND gates are universal).
\end{result*}

To prove the aforementioned result, we introduce a communication protocol for simulating geometrically-local circuits (see Lemma~\ref{Lemma: communication_protocol}). The classical hardness part of our result uses a direct product theorem about classical communication protocols from Jain and Kundu (see Fact~\ref{fact:jkleakageparallel}~\cite{jain2022direct}). We explore (see section~~\ref{sec: poq}) its application in providing verifiable quantum advantage~\cite{brakerski2020simpler,alagic2020non,kahanamoku2021classically,kalai2022quantum,jain2022direct}.
We also expand on our work, considering more complex circuits in higher dimensions (see subsection~\ref{sec: high_d_circuits}) and broader classes of Bell games (see subsection~\ref{subsec: wide_Bell games}). 

\subsection{Prior art}
Research on an unconditional separation between quantum and classical circuits was pioneered by Bravyi et al. \cite{bravyi2018quantum}. Their work rigorously demonstrated that constant-depth, bounded fan-in quantum circuits ($\mathrm{QNC}^0$) possess a verifiable computational advantage over constant-depth, bounded fan-in classical circuits ($\mathrm{NC}^0$). They introduced a relation problem that can be solved with certainty by a constant-depth quantum circuit employing geometrically-local gates on a two-dimensional grid of qubits. In contrast, any classical probabilistic circuit that achieves a success probability of at least $\frac{7}{8}$ in solving the same problem must have a depth that grows logarithmically with the input size. Importantly, the aforementioned separation also holds in the average-case scenario, where the classical circuit is required to solve a few instances of the problem randomly drawn from a suitable distribution. The foundation of this separation, as outlined in \cite{bravyi2018quantum}, stems from Bell nonlocality.

In a subsequent work, for a two-dimensional lattice of qubits, Ref.~\cite{coudron2021trading} discussed the existence of an efficiently samplable input distribution and demonstrated that, for any constant-depth classical circuit, the circuit returns an output that satisfies the relation with only a small probability on average over the choice of input. Furthermore, the success probability for the log-depth classical case (soundness error) in the relation problem described in \cite{coudron2021trading} can be exponentially small in the input size. In parallel, Ref.~\cite{gall2018average} provided an alternate proof for the average-case hardness using the framework of graph states. Both \cite{coudron2021trading} and \cite{gall2018average} employed parallel repetition to amplify the soundness guarantees.

Building upon the results of \cite{bravyi2018quantum}, Watts et al.~\cite{watts2019expo} extended the findings to the class of $\mathrm{AC}^0$ circuits, which represent classical, polynomial-size, constant-depth circuits comprising unbounded fan-in AND and OR gates, as well as NOT gates. They demonstrated that the relation problem introduced in \cite{bravyi2018quantum} is not in the $\mathrm{AC}^0$ class. Furthermore, they extended the results of \cite{bravyi2018quantum} to the average-case scenario and amplified the soundness guarantees, resulting in an exponentially small success probability for $\mathrm{AC}^0$ circuits.

Subsequently, Bravyi et al. \cite{bravyi2020quantum} expanded upon the findings of \cite{bravyi2018quantum}, this time also considering one-dimensional geometry. They proved that the separation between constant-depth classical and quantum circuits persists even in the presence of noise within the quantum circuits. The considered quantum circuits were subject to local stochastic noise, where a random Pauli error occurred at each time step in the ideal circuit. Although the error could affect multiple qubits, the probability of high-weight errors needed to be exponentially suppressed. To establish the separation in the noisy case, the authors employed techniques based on quantum error correction.

In Ref.~\cite{grier2020interactive}, the authors discussed a two-round interactive task that is solved by a constant-depth quantum circuit (using only Clifford gates, between neighbouring qubits of a 2D grid, with Pauli measurements), but such that any classical solution would necessarily solve $\oplus \mathrm{L}$-hard\footnote{Pronounced as Parity L. The task of simulating stabilizer circuits is $\oplus \mathrm{L}$-complete.} problems. The results in Ref.~\cite{grier2020interactive} were extended to the noisy case by~\cite{grier2021interactive} using techniques based on error correction from Ref.~\cite{bremner2011classical}.

\subsection{Comparison with prior work}
Here we delve into the similarities and differences between our work and the relevant literature (to the best of our knowledge).
\begin{itemize}
\item All previous works show a separation (for specific tasks) of constant v/s sub-logarithmic for depth required by general (geometrically non-local) quantum and classical circuits. We provide a much larger separation constant v/s sub-linear, however for geometrically local circuits. 
    \item Unlike Ref~\cite{coudron2021trading}, our results hold for $1D$ circuit families as well. \item Unlike Ref~\cite{bravyi2018quantum,bravyi2020quantum}, we provide noise robustness for every dimension and our noise robustness does not rely on error correction based approaches.

    \item Unlike any of the previous works, our result holds for wider class of Bell games.
    \item Unlike Ref~\cite{bravyi2018quantum,bravyi2020quantum}, the classical success probability is exponentially small.
    \item Unlike Ref~\cite{grier2020interactive,grier2021interactive}, our set-up is non-interactive.

    \item Unlike Ref~\cite{brakerski2020simpler,alagic2020non,kahanamoku2021classically,kalai2022quantum}, our quantum advantage protocol does not rely on any cryptographic assumption.

\end{itemize}
For a pictorial synopsis, please refer to the Table below. 

\begin{table}[H]
    \centering
    \small 
    \begin{tabular}{|c|c|c|c|c|c|c|c|c|c|}
    \hline
     & Depth (classical) & G-locality & Dimension & Bell-games & Soundness error & Robustness & QEC &  Qubits & Inputs \\ 
    \hline
    \cite{bravyi2018quantum} & sub-logarithmic & no & $D=2$ & specific & constant & no & N.A. & $?$ & avg${\star}$\\ 
    \hline
    \cite{coudron2021trading} & sub-logarithmic & no & $D=2$ & specific & 1/exp & yes & no & $10^6$ & avg\\ 
    \hline
    \cite{gall2018average} & sub-logarithmic & no & $D =2$ & no & 1/exp & yes & no & $?$ & avg\\ 
    \hline
    \cite{watts2019expo} & sub-logarithmic & no & $D = 2$ & no & 1/exp & no & N.A. & $?$ & avg\\ 
    \hline
    \cite{bravyi2020quantum} & sub-logarithmic & no & $D \in \{1, 3\}$ & specific & constant & yes~${\star\star}$ & yes & $?$ & avg\\ 
    \hline
    This work & sub-linear & yes & $D \in \mathbb{N}$ & general & 1/exp & yes & no & $\star \star \star$ & avg\\ 
    \hline
    \end{tabular}
    \label{tab:comparision}
\end{table}

\noindent $\star:$ Initial results were shown for the worst case, subsequently improved to the average case.\\
\noindent $\star \star$: The robustness was shown for $D=3$ case. \\
$\star \star \star$ Depending on the extent to which geometric locality is enforced by nature on real-world classical circuits. 
\newline 

Here, the table on the previous page juxtaposes our research with relevant literature on unconditional separation in a non-interactive set-up, underscoring both similarities and differences. The first column indicates the depth of the classical circuit that remains unable to solve some specific  relation as per the various works listed in each row. Yet, for each of these relations, a quantum circuit of constant depth suffices. The next column discusses if these studies are predicated on the assumption of geometric locality for classical circuits. Note that prior studies have demonstrated a depth separation of constant versus sub-logarithmic for tasks between general (geometrically non-local) quantum and classical circuits. In contrast, we present a more substantial separation of constant versus sub-linear depth, but specifically for geometrically local circuits. The  third column, labelled ``Dimension'', indicates the circuit family (geometric) dimension where an unconditional separation has been established in various studies. The subsequent column discusses whether the proof uses ideas based on Bell games. Here ``specific'' means that the proof works for a particular class of Bell games. Our proof works for general Bell games.  The fifth column presents the classical success probability (also referred to as the soundness error) as a function of input size $n$. The sixth column addresses the presence or absence of noise robustness in different works. The next column explores if techniques rooted in quantum error correction were employed to achieve this robustness. The penultimate column presents the number of qubits required for a possible demonstration of experimental quantum advantage. The final column discusses the input distribution for which the result holds (average case versus worst case). While question marks (?) indicate that no value has been reported or is available, N.A. denotes ``not applicable.''

\subsection{Why geometric locality matters}  \label{subsec: geometric_locality}
Geometric locality refers to the idea that an object's influence is limited to its immediate surroundings. A theory that follows this principle is called a ``local theory'' and differs from the concept of instant action over large distances. Locality emerged from field theories in classical physics, suggesting the need for an intermediary element between two points for one location to affect another. This intermediary, such as a wave or particle, travels through space connecting the points to transmit the influence. 

The concept of geometric locality holds significant importance in classical memory devices as well, primarily because programs have a tendency to access data that is in close spatial proximity to other previously accessed data~\cite{kumar1998exploiting}. In this context, the term ``nearby'' pertains to the spatially surrounding memory addresses of the data. Moreover, note that transistors are locally compact objects. To execute geometrically non-local gates, one may require multiple transistors or a powerful signal that can travel large distance. There exists a balance between signal strength and the efficacy of adjacent transistors, as potent signals can cause undesired interactions. Consequently, there's a maximum limit to signal strength, potentially enforcing the use of several transistors when applying geometrically non-local gates. This increased transistor count consequently results in a time cost when implementing such gates. The speed of light is another constraint for signals to travel long distances. As asserted by Einstein, geometric locality is fundamental to our understanding of the physical realm (both classical and quantum).

\subsection{Need for the initial (non-local) resource state} \label{subsec: initial_resource}
The study of computational complexity often involves the consideration of input-independent resource states, randomness being an example. While efficient deterministic algorithms remain challenging for many problems, randomized algorithms offer elegant and efficient solutions by utilizing independent and unbiased coin tosses during computation, resulting in random variable outputs. Despite the possibility of errors, embracing perfect randomness as additional input and accommodating small errors in the output can significantly enhance efficiency for difficult problems. This approach has given rise to complexity classes like $\mathrm{BPP}$ and $\mathrm{FBPP}$ that leverage randomness as an initial resource~\cite{arora2009computational}.
In the context of Bell non-locality, classical provers can utilize shared randomness, while in the quantum scenario, entangled states serve as the analogous shared resource~\cite{bell1964einstein,clauser1969proposed,brunner2014bell}. Motivated by the aforementioned arguments, in this work, we allow classical circuits to start with input-independent (geometrically-non-local) shared randomness. For quantum circuits, quantum circuits can begin with an input independent (geometrically-non-local) entangled state.

\subsection{Organization}
We start with the preliminaries in section~\ref{sec:prelim}. In this part, we talk about parallel repetition of Bell games and the parallel repetition result (see Fact~\ref{fact:jkleakageparallel}) from Jain and Kundu~\cite{jain2022direct}, that we use to prove our classical hardness result. Moving on to section~\ref{sec: main_results_sec}, we discuss our results. In section~\ref{subsec: relation}, we introduce the relation problem based on parallel repetition of the Magic Square game. In section~\ref{subsec: com_protocol}, we present a two-party communication protocol for simulating geometrically-local circuits (see Lemma~\ref{Lemma: communication_protocol}). The main result is presented in section~\ref{subsec: main_result} under Theorem~\ref{thm: main}. In section~\ref{sec: poq}, we explore its application for verifiable quantum advantage. Moving further, in section~\ref{sec: generamization_applications}, we expand on our work, considering more complex circuits in higher dimensions and broader classes of Bell games.  Finally, in section~\ref{sec: discuss}, we touch upon unresolved questions.

\section{Preliminaries} \label{sec:prelim}

\subsection{Parallel two-prover Bell games}

In a two-prover one-round Bell game, denoted as $\mathcal{G}$, participants include a referee and two isolated players, referred to as Alice and Bob. Notably, Alice and Bob communicate exclusively with the referee, without any direct interaction between them. The game's process can be described as follows.

The referee initiates the game by randomly selecting a question pair $(x, y)$ from a specific distribution $\pi$. Subsequently, the referee sends question $x$ to Alice and question $y$ to Bob. Both Alice and Bob provide their respective answers, denoted as $a$ and $b$, respectively. The game's objective is to determine a winning scenario by evaluating the predicate $V(x, y, a, b)$. A victory is declared when the predicate evaluates to $1$, indicating a successful outcome based on the combination of questions and players' responses. We proceed with the formal Definition of a two-prover one-round Bell game.
\begin{definition}
[Two-prover one-round Bell game] Given a predicate $V:\mathcal{X}\times\mathcal{Y\times\mathcal{A}}\times\mathcal{B}\rightarrow\{0,1\}$
and a probability distribution $\pi$ on $\mathcal{X}\times\mathcal{Y}$,
a two-prover one-round Bell game $\mathcal{G}=\left(V,\pi\right)$ involves two provers
and one verifier, which proceeds as follows.
\begin{enumerate}
\item The verifier samples a pair of questions $\left(x,y\right)\in\mathcal{X}\times\mathcal{Y}$
according to the probability distribution $\pi$.
\item The verifier sends $x$ and $y$ to the two provers and receives answers
$a\in\mathcal{A}$ and $b\in\mathcal{B}$ respectively.
\item The verifier applies the predicate $V:\mathcal{X}\times\mathcal{Y\times\mathcal{A}}\times\mathcal{B}\rightarrow\{0,1\}$
and accepts the answers if the outcome is $1$, rejects otherwise.
\end{enumerate}
\end{definition}

In the game $\mathcal{G}$, the maximum probability of winning for Alice and Bob depends on the resources they are allowed to utilize. When their answers are solely determined by the received question (perhaps with some shared randomness), we refer to this maximum winning probability as the classical value, represented as $\mathrm{val}\left(\mathcal{G}\right)$ or $\omega_c\left(\mathcal{G}\right)$. However, in the realm of quantum mechanics, Alice and Bob can leverage entanglement to establish correlations that cannot be achieved solely through shared randomness. In such cases, the maximum winning probability is referred to as the quantum value, denoted as  $\mathrm{val}^{\star}\left(\mathcal{G}\right)$ or $\omega_q\left(\mathcal{G}\right)$.

In the game $\mathcal{G}^m$, Alice and Bob engage in $m$ separate and independent instances of game $\mathcal{G}$ simultaneously. One may inquire about the value of a game $\mathcal{G}^m$. This concept is formally known as parallel repetition, capturing the essence of playing multiple instances in parallel.

\begin{definition}
[Parallel repetition] Given a two prover game $\mathcal{G}=\left(V,\pi\right)$ and
some $m\in\mathbb{N},$ an $m$-fold parallel repetition of the game
$\mathcal{G}$, denoted by $\mathcal{G}^{m}$ involves following:
\end{definition}

\begin{enumerate}
\item The verifier samples $m$ pairs of questions, say $\left(x_{1},y_{1}\right),\left(x_{2},y_{2}\right)\cdots\left(x_{m},y_{m}\right)$
independently.
\item The sequence $\left(x_{1},\cdots,x_{m}\right)$ is sent to the
first prover. The second prover receives $\left(y_{1},\cdots,y_{m}\right).$
\item The first prover sends the answer $\left(a_{1},a_{2},\cdots,a_{m}\right).$
The second prover sends $\left(b_{1},b_{2},\cdots,b_{m}\right).$ 
\item The verifier accepts iff $\forall i\in[m]: V\left(x_{i},y_{i},a_{i},b_{i}\right)=1.$
\end{enumerate}
\begin{remark}
It is easy to see that $\mathrm{val}\left(\mathcal{G}^{m}\right)\geq$ $\mathrm{val}\left(\mathcal{G}\right)^{m}.$
This is because the combined strategy to win the game includes the
independent single-shot strategies as a special case.
\end{remark}

\subsection{Magic Square game}
The Magic Square game is an example of a two-prover Bell game $\mathcal{G}_{\mathrm{MS}}$ where $\mathrm{val}^{\star}\left(\mathcal{G}_{\mathrm{MS}}\right) = 1$ is strictly greater than  val$\left(\mathcal{G}_{\mathrm{MS}}\right)= \frac{8}{9}$. In this game, Alice and Bob receive random inputs from the set $\{0, 1, 2\}$. They receive these inputs independently and uniformly. They produce outputs $a$ and $b$ as binary strings of length three, with the condition that the bitwise XOR (exclusive OR) of the elements in $a$ should result in 0, and the XOR of the elements in $b$ should result in 1. Their goal is to satisfy the condition $a[y] = b[x]$, where $x$ and $y$ are their respective inputs. In other words, they aim to make sure that the $y$th element of $a$ matches the $x$th element of $b$.
Formally, the Magic Square game $\mathcal{G}_{\mathrm{MS}} = (\pi, V)$  for $V: X \times Y \times A \times B \rightarrow \{0,1\} $ can be described as follows.
\begin{enumerate}
    \item The input sets are $X = Y = \{0, 1, 2\}$.
    \item The output sets are $A = B = \{0,1\}^3$.
    \item The probability distribution $\pi$ is defined such that $\pi(0, 0) = \pi(0, 1) = \pi(1, 0) =\cdots =  \pi(2, 2) = \frac{1}{9}$, meaning each pair of inputs is equally likely.
    \item The predicate  $V(x,y,a,b)= 1$ iff $a[0] \oplus a[1] \oplus a[2] = 0$, $b[0] \oplus b[1] \oplus b[2] = 1$ and $a[y] = b[x]$.
\end{enumerate}
See Figure~\ref{fig:MS_Game} for a pictorial depiction of the Magic Square game. To prove our classical hardness statement, we will use the following Fact from~\cite{jain2022direct}.
\begin{fact}[\cite{jain2022direct}] \label{fact:jkleakageparallel}
    The probability (averaged over the inputs drawn uniformly) of classical players winning $0.99 n$ games out of $n$ parallel copies of the Magic Square game, with $o(n)$ bits being interactively leaked between Alice and Bob is at most $2^{-\Omega(n)}$.
\end{fact}

\begin{figure}
\centering
\includegraphics[width=0.7\textwidth, height=7cm]{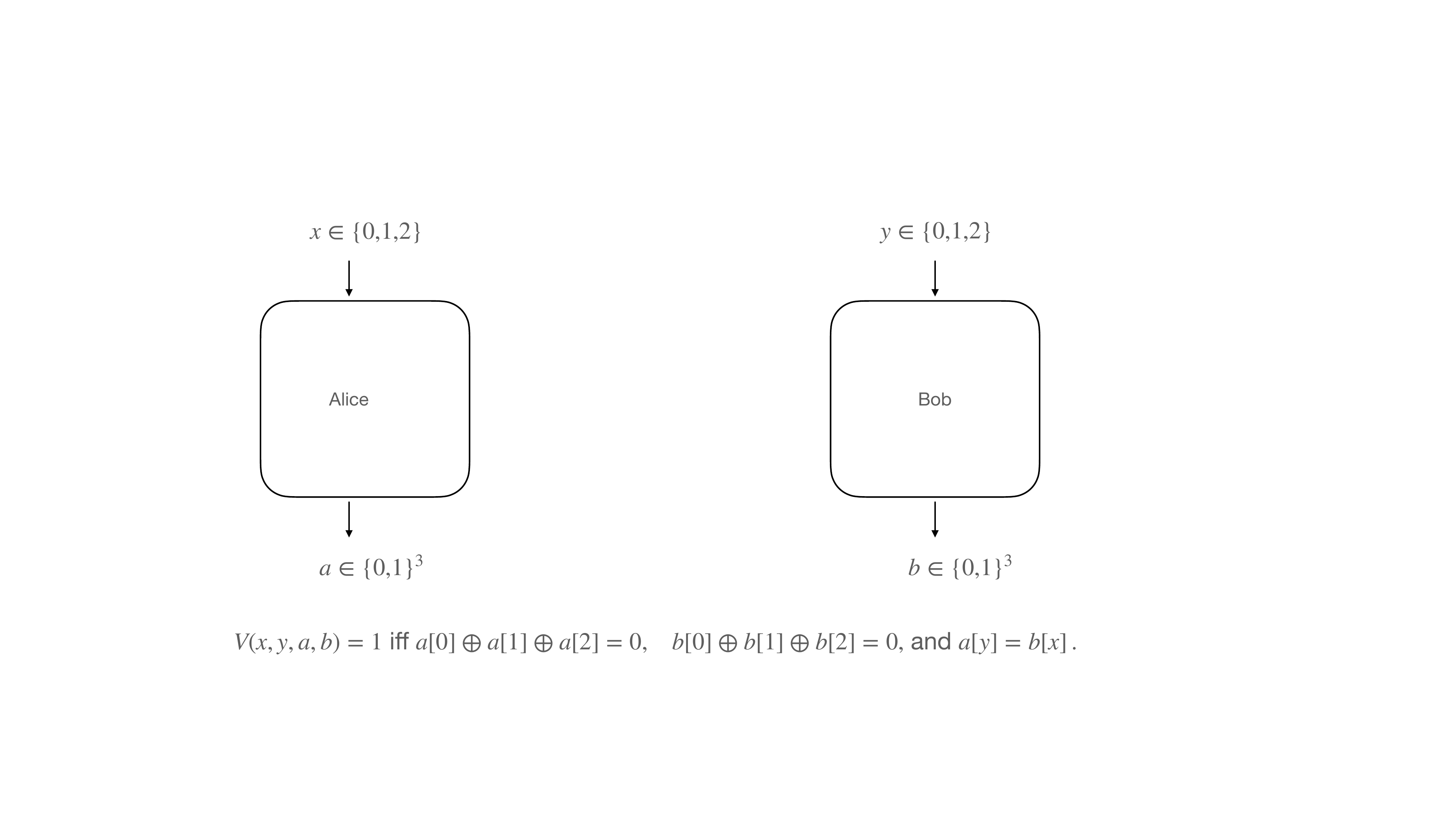}
\caption{The Magic Square game.}
\label{fig:MS_Game}
 \end{figure}

\begin{figure}
\centering
\includegraphics[width=0.7\textwidth, height=5cm]{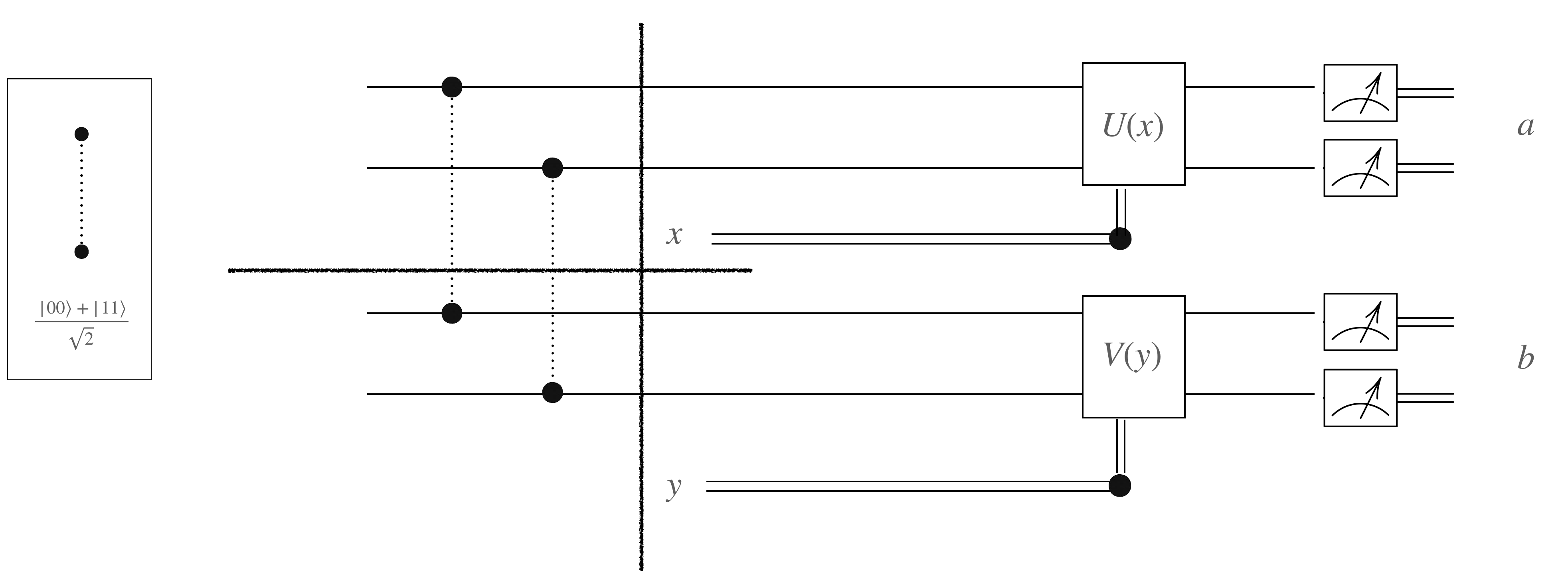}
\caption{The Magic Square game circuit. There are two cuts, represented by dashed lines. The computation starts right to the vertical dashed line. The quantum state prepared left to the vertical cut is an input-independent initial resource state (pair of Bell-pairs). For a detailed description of the \( U \) and \( V \) gates for different inputs, please consult Figure~2 in Ref~\cite{bravyi2020quantum}.
}
\label{fig:MS_circuit}
 \end{figure}

\section{Constant-depth quantum versus sublinear-depth classical circuits with geometric locality} \label{sec: main_results_sec}

\subsection{The relation problem: parallel repetition of the Magic Square game} \label{subsec: relation}
In the realm of quantum computing, notable instances have emerged wherein researchers attained significant insights by shifting their focus from decision problems to relation problems. This transition has led to noteworthy achievements in various domains, including quantum supremacy protocols such as Boson sampling~\cite{aaronson2011computational,zhong2020quantum} and random circuit sampling~\cite{boixo2018characterizing,arute2019quantum}. Moreover, recent results like the Yamakawa-Zhandry protocol~\cite{yamakawa2022verifiable}, alongside several unconditional complexity theory separations~\cite{aaronson2023qubit,grier2020interactive,grier2021interactive,coudron2021trading,bravyi2018quantum,bravyi2020quantum}, exemplify the fruitful outcomes of this shift in attention. Relational problems, denoted by a relation $R \subseteq \{0, 1\}^* \times \{0, 1\}^*$, entail finding a suitable output $y$ for a given input $x$ such that $(x, y) \in R$. 
\begin{definition} \label{defline_input} (Line input of size $n$)  We are given inputs $x, y\in  \{0,1\}^{2n}$ arranged along a line such that the following are adjacent: (1) $x_i$ and $x_{i+1}$ for $i \in \{1,2,\cdots, n-1\}$ (2) $y_i$ and $y_{i+1}$ for $i \in \{1,2,\cdots, n-1\}$, and (3) $x_n$ and $y_1$.
\end{definition}
We define the relational problem with geometrical locality constraint corresponding to parallel repetition of $n$ Magic Square games, henceforth referred as PARMAGIC$(\delta)$.
\begin{definition} [PARMAGIC $(\delta)$] Given line input $x, y\in  \{0,1\}^{2n}$, the goal is to obtain $a, b\in  \{0,1\}^{2n}$ such that $n(1-\delta)$ tuples $(x_i,y_i,a_i,b_i)$ satisfy the Magic Square game predicate $V(x_i,y_i,a_i,b_i)$.
\end{definition}

\begin{figure}
\centering
\includegraphics[width=0.9\textwidth, height=10cm]{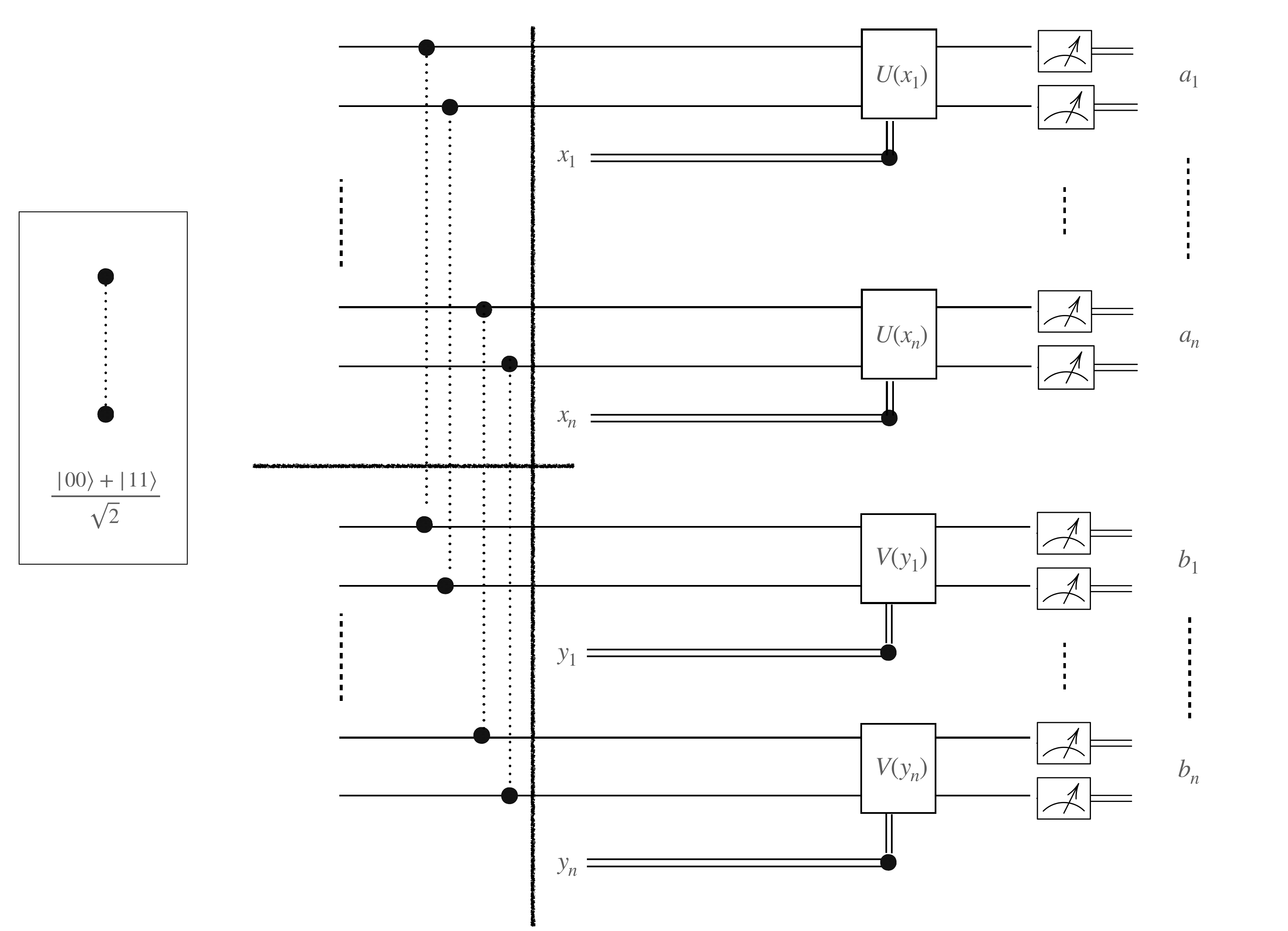}
\caption{The PARMAGIC circuit. There are two cuts, represented by solid black lines. The computation starts right to the vertical solid black line. The quantum state prepared left to the vertical cut is an input-independent initial resource state. The vertical line connecting two solid black circles represents a two-qubit maximally entangled state. The gates $U(x_i)$ and $V(y_j)$ are same as the the two qubit classically controlled gates $U$ and $V$ in Figure~\ref{fig:MS_circuit}.}
\label{fig:PARMAGIC_circuit}
 \end{figure}

\subsection{Communication protocol for simulating geometrically-local circuits} \label{subsec: com_protocol}
A geometrically-local, randomness-assisted, depth $d$ classical circuit $\mathscr{C}$ with constant fan-in gates can be simulated using a two-party randomness-assisted classical communication protocol. This simulation process incurs a communication of $\mathcal{O}(d)$ bits, indicating that the amount of information exchanged scales linearly with the depth of the circuit. We proceed with a rigorous treatment of the statement mentioned above. 
\begin{figure}
\centering
\includegraphics[width=0.8\textwidth, height=8cm]{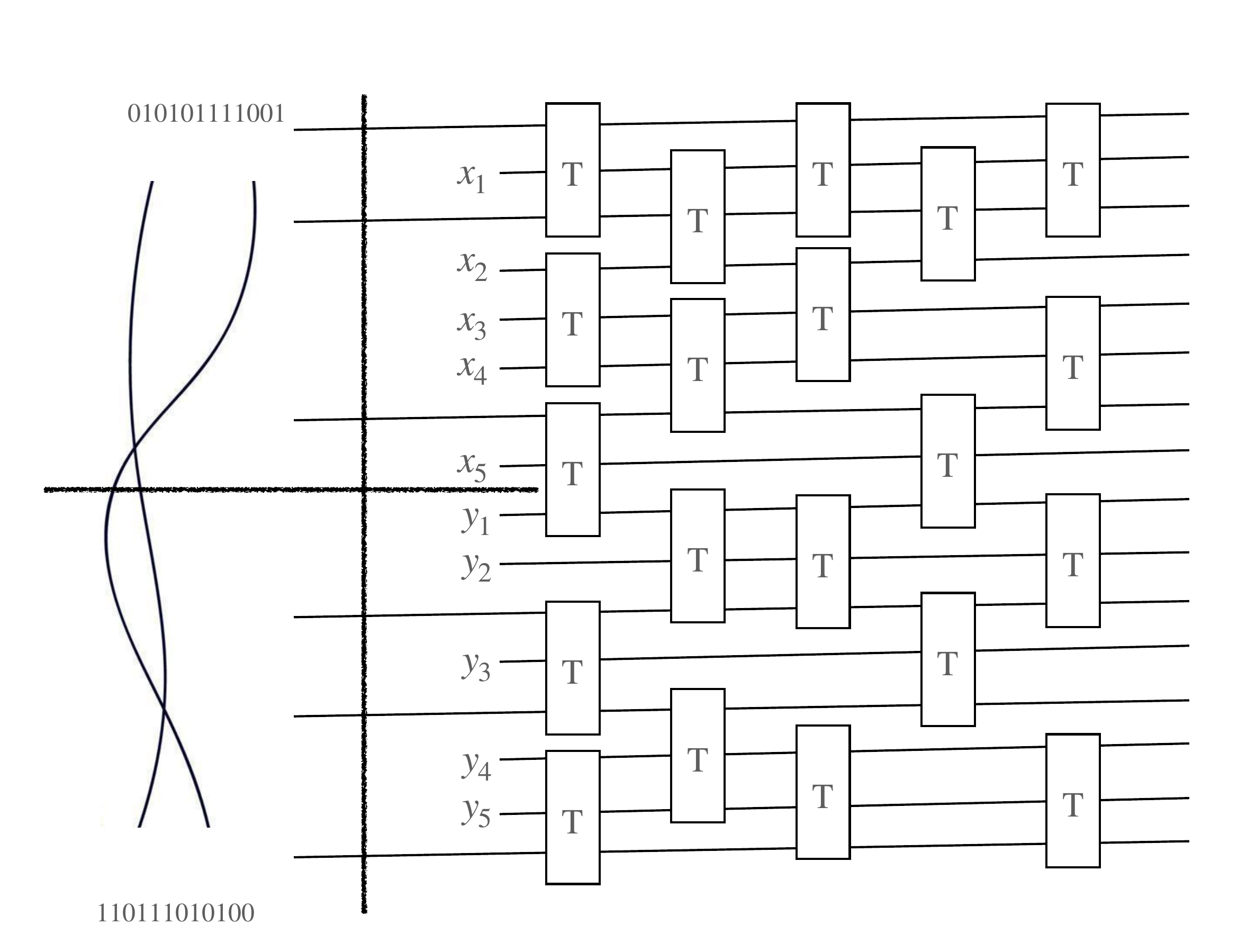}
\caption{An example classical circuit. Just like in the quantum case, we have two cuts represented by dashed lines. The content to the left of the vertical cut represents the shared randomness.}
\label{fig:circuit_classical}
 \end{figure}

\begin{claim} \label{claim:Toffoli} A $1D$, depth $d$, geometrically-local, classical circuit consisting of fan-in $2$ NAND gates can be simulated using a $1D$, $O(d)$-depth, geometrically-local, classical circuit consisting of fan-in $3$, fan-out $3$, Toffoli gates.
    
\end{claim}
\begin{proof}
The Toffoli gate operates as $\{a,b,c\} \rightarrow \{a,b,c \oplus ab\}$
where $\oplus$ represents bitwise modulo addition. With $c$ set to 1, it simulates the NAND gate as $\{a,b,1\} \rightarrow \{a,b, 1 \oplus ab\}$.  

Furthermore, Toffoli gate can emulate the fan-out operation: $\{a,1,0\} \rightarrow \{a,1, a\}$. Note that in a  $1D$, geometrically-local circuit, the fan-out of any NAND gate is limited to $3$. Hence, when leveraging the Toffoli gate for fan-out simulation, each layer incurs a constant depth cost. This shows the desired.
\end{proof}
For a $1D$ circuit with only Toffoli gates, we define a notion of horizontal cut for any line input of size $n$.
\begin{definition} \label{defn: h-cut}(Horizontal cut) Let $\mathscr{C}$ be a $1D$,  geometrically-local circuit of depth $d$, consisting of Toffoli gates, acting on line input of size $n$. For each $i\in[d]$, there exists a partition of the input wires to the gates in layer $i$ in two disjoint sets $U_{i}$ and $D_{i}$ using the following rule.
\begin{enumerate}
    \item $U_1$ consists of the input wires corresponding to $(x_1, \ldots, x_n)$ and $D_1$ consists of the input wires corresponding to $(y_1, \ldots, y_n)$.
    \item If a Toffoli gate T in $\mathscr{C}$ receives all inputs from $U_i$, all output wires of T are mapped to $U_{i+1}$.
     \item If a Toffoli gate T in $\mathscr{C}$ receives all inputs from $D_i$, all output wires of T are mapped to $D_{i+1}$.
      \item If a Toffoli gate T in $\mathscr{C}$ receives inputs from $U_i$ as well as $D_i$, all output wires of T are mapped to $U_{i+1}$.
\end{enumerate}
\end{definition}
We proceed with the formal statement on the communication complexity of our communication protocol for simulating geometrically-local circuits.
\begin{lemma} \label{Lemma: communication_protocol}
    A $1D$, depth $d$ geometrically-local circuit $\mathscr{C}$, consisting of Toffoli gates, acting on line input of size $n$, using shared randomness as an initial resource can be simulated using a randomness-assisted two-party communication protocol with communication $\mathcal{O}(d)$.
\end{lemma}
\begin{proof} Let us denote the set of gates in $\mathscr{C}$ present in layer $i$  as $\mathscr{G}_{i}$.  The set  $\mathscr{G}_{i}$ is the union of three disjoint sets:
\begin{enumerate}
    \item $\mathscr{G}^u_{i}$: the set of gates that receive input wires from $U_i$, 
    \item $\mathscr{G}^d_{i}$: the set of gates that receive input wires from $D_i$ , and 
    \item $\mathscr{G}^a_{i}$: the set of gates that receive input wires from $U_i$ as well as $D_i$.
\end{enumerate} 
  In the randomness-assisted, communication protocol $\mathscr{P}$, Alice receives inputs  $(x_1, \ldots, x_n)$, while Bob receives inputs $(y_1, \ldots, y_n)$. They start with the same shared randomness as is present in $\mathscr{C}$. Given that Alice and Bob have already simulated the circuit $\mathscr{C}$ up to layer $i-1$, Alice can simulate locally the gates in  $\mathscr{G}^u_{i}$ and Bob can simulate locally the gates in  $\mathscr{G}^d_{i}$ without any communication. To simulate a gate in $\mathscr{G}^a_{i}$, Bob sends to Alice (using $O(1)$ bits) the value of the variables corresponding to the input wires in $D_i$. Note that, since $\mathscr{C}$ is $1D$, geometrically-local, there is at most one gate in $\mathscr{G}^a_{i}$.  Hence overall communication in $\mathscr{P}$ is $O(d)$. This completes the proof.
\end{proof}

\subsection{Main result} \label{subsec: main_result}
\begin{theorem} \label{thm: main}
Let $\delta \in [0,0.1]$. For the $\mathrm{PARMAGIC(\delta)}$ problem with line-input of size $n$, the following holds.
\begin{enumerate}
    \item (Completeness) $\mathrm{PARMAGIC(\delta)} \in \mathrm{\{gl\}FQU}\left(2, \delta/100, 1-2^{-\Omega(n)} \right)$, 
    
    where $ \mathrm{\{gl\}FQU}\left(d, \eps, p \right)$ represents the class of relation problems that can be solved by $1D$, geometrically-local, depth $d$,  uniform quantum circuits with  fan-in $4$ gates, per gate noise  $\eps$ and probability of success at least $p$ (worst-case over inputs).
    \item (Soundness) $\mathrm{PARMAGIC(\delta)} \notin \mathrm{\{gl\}F}\left(o(n), 2^{-o(n)} \right)$, 
    
    where 
    $\mathrm{\{gl\}F}\left(d, p \right)$ represents the class of relation problems that can be solved by $1D$,  geometrically-local, depth $d$, noiseless, non-uniform, classical circuits with fan-in $2$ NAND gates, with success probability at least $p$ (averaged over the inputs drawn uniformly). 
\end{enumerate}
\end{theorem}
\begin{proof} 

\emph{(1) Completeness}: Use the circuit in Figure~\ref{fig:PARMAGIC_circuit}. The quantum state prepared left to the vertical solid line is the initial resource state prepared in auxiliary registers. The depth $2$, fan-in $4$, quantum circuit, assuming noiseless gates, solves $\mathrm{PARMAGIC(0)}$ with success probability $1$. It is easily seen that this circuit family is uniform. Thus,  $$\mathrm{PARMAGIC(0)}  \in \mathrm{\{gl\}FQU}\left(2, 0 , 1\right).$$ 

Now let us assume that per gate noise is $\delta/100$. Since the depth of the circuit is $2$, it is not difficult to argue that any fixed input-output pair $(x_i,y_i,a_i,b_i)$ will satisfy the Magic Square predicate with probability at least $1 - \delta/2$. Thus, using standard concentration bounds, at least $(1 - \delta)$ fraction of the input-output pairs will satisfy the Magic Square predicate with probability at least $1 - 2^{-\Omega(n)}$, implying,
$$\mathrm{PARMAGIC(\delta)} \in \allowbreak \mathrm{\{gl\}FQU}\left(2, \delta/100, 1-2^{-\Omega(n)} \right).$$

\emph{(2) Soundness}: Follows directly from Claim~\ref{claim:Toffoli},  Lemma~\ref{Lemma: communication_protocol} and  Fact~\ref{fact:jkleakageparallel}.
\end{proof}
    
\section{Application: Verifiable quantum advantage} \label{sec: poq}
In this two-party protocol, the verifier $\verifier$ is an efficient (uniform) classical circuit and the honest prover $\prover$ is a constant depth geometrically-local (uniform) quantum circuit.

\begin{center}
\framebox{ 
\begin{minipage}{0.8\textwidth}
\vspace{5pt}
\textbf{Verifiable quantum advantage  protocol}
\vspace{5pt}
\hrule
\vspace{5pt}

\begin{enumerate}
\item  $\verifier$ chooses $x_1, \ldots, x_n \in \{0, 1, 2\}^n$ and $y_1, \ldots, y_n \in \{0, 1, 2\}^n$ uniformly at random.
\item $\verifier$ \textbf{simultaneously} sends $\textvisiblespace,\textvisiblespace,x_1,\textvisiblespace,\textvisiblespace,x_2,\textvisiblespace,\textvisiblespace, \ldots,\textvisiblespace,\textvisiblespace, x_n,\textvisiblespace,\textvisiblespace, y_1,\textvisiblespace,\textvisiblespace, y_2,\textvisiblespace,\textvisiblespace, \\ 
\ldots,\textvisiblespace,\textvisiblespace, y_n$ to $\prover$, in geometrically-local order. Here, geometrically-local order means that we have a line input (see Definition~\ref{defline_input}) with a $\textvisiblespace,\textvisiblespace$ before every element of the input string.
\item $\prover$ follows the protocols as in Figure~\ref{fig:input_swap} followed by in Figure~\ref{fig:PARMAGIC_circuit} and sends $a_1, \ldots, a_n, b_1, \ldots, b_n$ to $\verifier$.
\item $\verifier$ outputs $\top$ if $(a_i[0] \oplus a_i[1] \oplus a_i[2] = 0) \land (b_i[0] \oplus b_i[1] \oplus b_i[2] = 1) \land (a_i[y_i] = b_i[x_i])$ for at least $(1 - \delta)n$ many $i$'s in $[n]$, and outputs $\bot$ otherwise.
\end{enumerate}

\end{minipage}}
\end{center}
The completeness and soundness of our protocol follows from Theorem~\ref{thm: main}. 

\begin{figure}[H]
\centering
\includegraphics[width=0.65\textwidth]{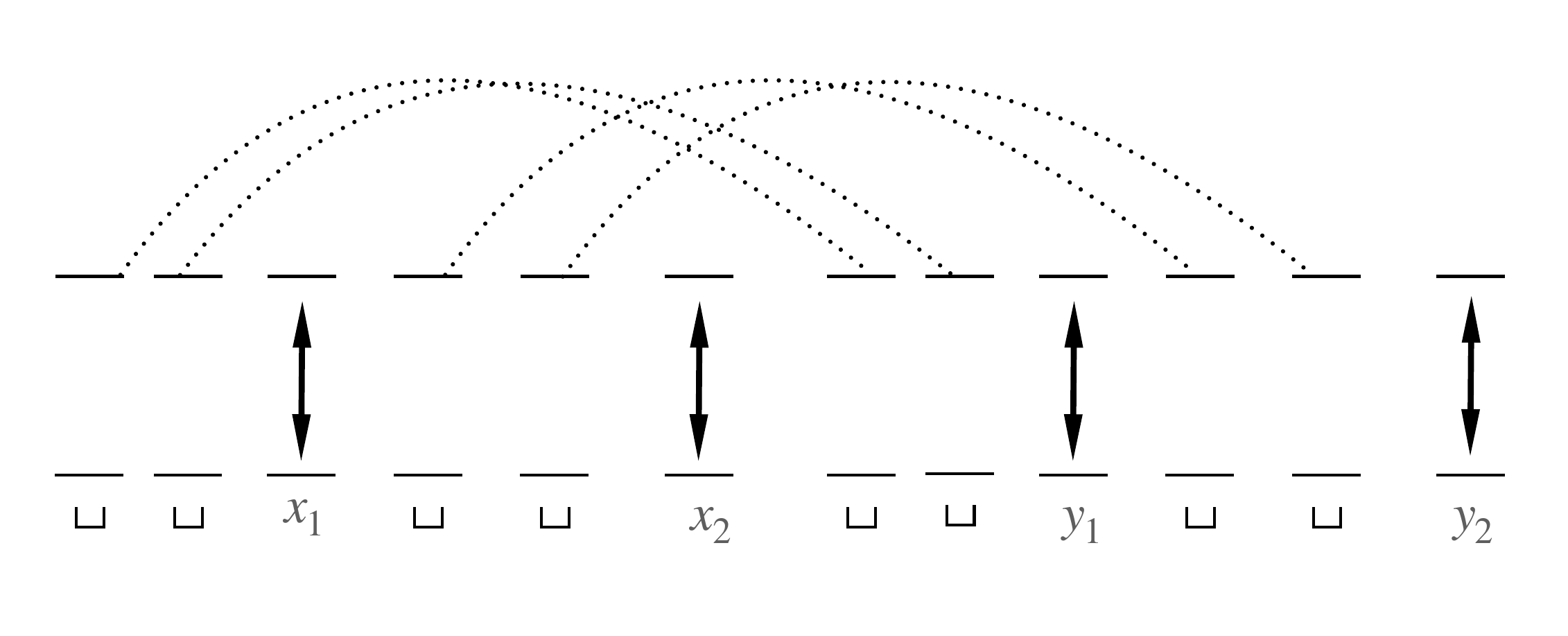}
\caption{The top vertical set of line segments denote the input register to the PARMAGIC circuit with $n =2$. The line segments connected by dashed curves lines represent the Bell pair $\ket{\phi_{+}}= \frac{\ket{00} + \ket{11}}{\sqrt{2}}$, the initial resource state prepared in auxiliary registers of the PARMAGIC circuit.  The value in every $3k$-th lower input register $\forall k \in \{1,2,3,4\}$ is supplied as input to the PARMAGIC cicuit by swapping the lower input register value with the corresponding upper input register. The honest quantum prover ignores the input registers containing  $\textvisiblespace$.}
\label{fig:input_swap}
 \end{figure}

The fastest classical processors operate at a clock speed of 
 a few GHz (upper bounded by 10 GHZ)~\cite{coudron2021trading}. This means that a circuit of depth $d$ requires approximately \( 10^{-9}d \) seconds to execute. In comparison, the current gate operation times for photonic quantum computers are roughly a few nanoseconds (upper bounded by 2)~\cite{gerbert2018next}. To achieve a quantum advantage, against geometrically-local classical circuits, we need to ensure that \( 10^{-10} d > 2 \times 10^{-9} \). When \( d>20 \), this inequality holds true. Therefore, it may be possible to demonstrate quantum advantage using the technology available in the NISQ era. This will also depend on how much geometric locality is enforced by nature on the real-world classical circuits.

\section{Generalizations} \label{sec: generamization_applications}
\subsection{Higher dimensional circuits} \label{sec: high_d_circuits}
We start with defining the analogue of Definition~\ref{defline_input} for any $D \in \mathbb{N}$ (although of course we note that there are three accessible spatial dimensions). 
\begin{definition} [$D$-dimensional grid input of size $n$ and orientation graph $G= (V,E)$]  We are given inputs $x, y\in  \{0,1\}^{2n}$ arranged on a $D$-dimensional grid with orientation graph $G$ such that the following are adjacent independent of $G$:
 (1) $x_i$ and $x_{i+1}$ for $i \in \{1,2,\cdots, n-1\}$ (2) $y_i$ and $y_{i+1}$ for $i \in \{1,2,\cdots, n-1\}$, (3) $x_i$ and $y_i$ for $i \in \{1,2,\cdots,n^\frac{D-1}{D}\}$.
 Furthermore, $u,v \in V \times V$ are adjacent whenever $(u,v) \in E$. Here, $V=  \{x_i\}_{i=1}^{n} \cup  \{y_i\}_{i=1}^{n}$.
\end{definition}

\begin{figure}[H]
\centering
\includegraphics[width=0.55\textwidth]{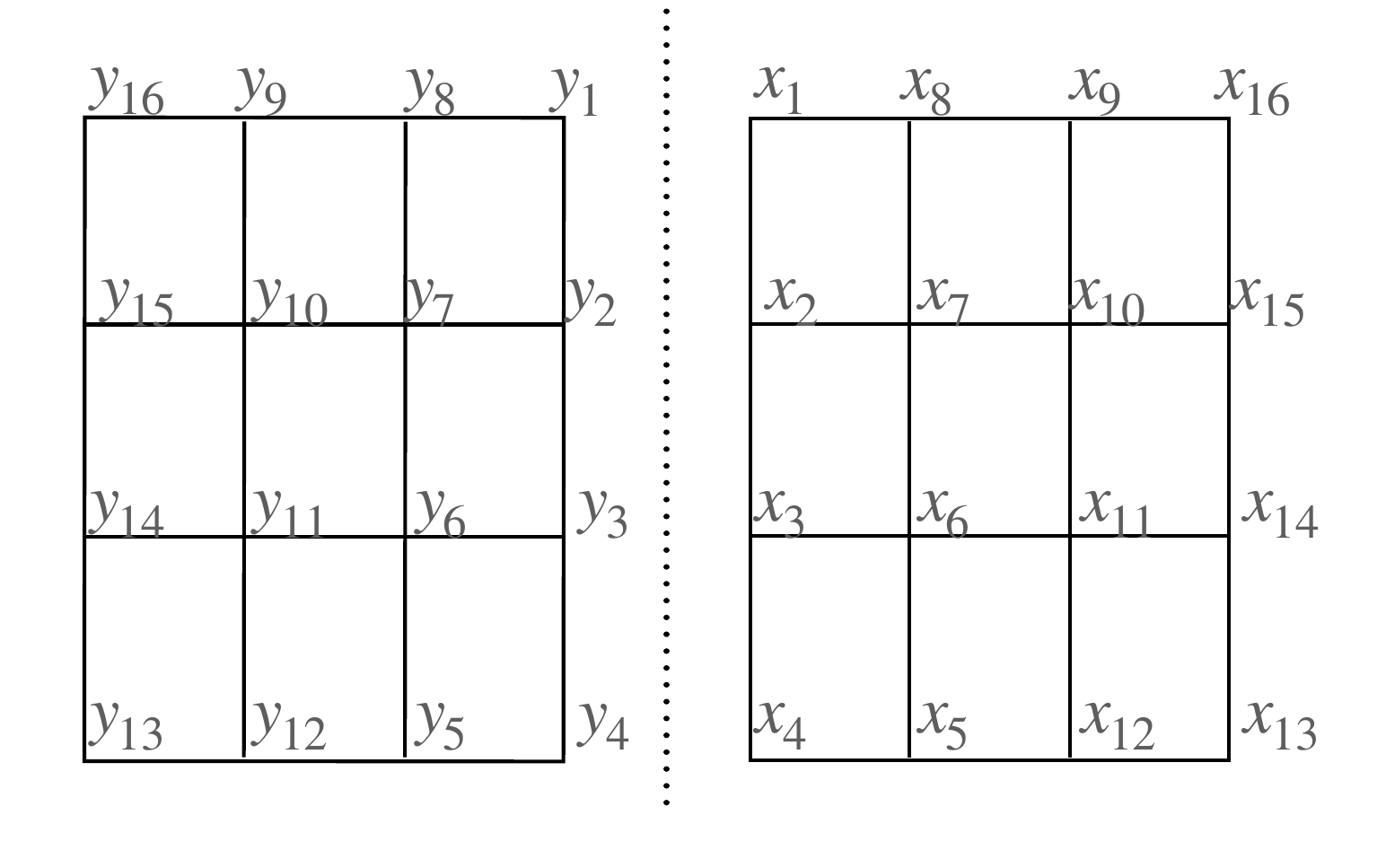}
\caption{The $D=2$ dimensional grid input of size $n=16$. There are four pairs of $x_i$'s and $y_j$'s, which are adjacent. In general, there are $n^{\frac{D-1}{D}}$  pairs of $x_i$'s and $y_j$'s which are adjacent, for $D$-dimensional grid input of size $n$.}
\label{fig:2D_grid}
 \end{figure}

Now we define the $D$-dimensional relational problem with geometrical locality constraint corresponding to the parallel repetition of $n$ Magic Square games, henceforth referred as PARMAGIC$(\delta, D)$. Since our result does not depend on the orientation graph $G$, we will drop $G$ hereafter.
\begin{definition} [PARMAGIC $(\delta,D)$] Given a $D$-dimensional grid input $x, y\in  \{0,1\}^{2n}$, the goal is to obtain $a, b\in  \{0,1\}^{2n}$ such that $n(1-\delta)$ tuples $(x_i,y_i,a_i,b_i)$ satisfy the Magic Square game predicate $V(x_i,y_i,a_i,b_i)$.
\end{definition}
Using PARMAGIC $(\delta,D)$, we have the following result in higher dimensions.

\begin{theorem} \label{thm: high_d}
Let $\delta \in [0,0.1]$. For the $\mathrm{PARMAGIC(\delta, D)}$ problem corresponding to $D$-dimensional grid-input of size $n$, where $D \in \mathbb{N}$, the following holds:
\begin{enumerate}
    \item (Completeness) $\mathrm{PARMAGIC(\delta, D)} \in \mathrm{\{gl\}FQU}\left(2, \delta/100, 1-2^{-\Omega(n)}, D \right)$, 
    
    where $ \mathrm{\{gl\}FQU}\left(d, \eps, p, D \right)$ represents the class of relation problems that can be solved by D-dimensional, geometrically-local, depth $d$,  uniform quantum circuits with  fan-in $4$ gates, per gate noise  $\eps$ and probability of success at least $p$ (worst-case input).
    \item (Soundness) $\mathrm{PARMAGIC(\delta, D)} \notin \mathrm{\{gl\}F}\left(o(n^\frac{1}{D}), 2^{-o(n)}, D \right)$, 
    
    where 
    $\mathrm{\{gl\}F}\left(d, p, D \right)$ represents the class of relation problems that can be solved by D-dimensional,  geometrically-local, depth $d$, noiseless, non-uniform, classical circuits with fan-in $2$ NAND gates, with success probability at least $p$ (averaged over the inputs drawn uniformly). 
   
\end{enumerate}
\end{theorem}
\begin{proof}
The proof of completeness is similar to the completeness proof of Theorem~\ref{thm: main}. For soundness, we can use the analogue of Lemma~\ref{Lemma: communication_protocol} for higher dimensional circuits. In this case, the corresponding communication protocol may communicate per layer $\mathcal{O}( n^\frac{D-1}{D})$ bits. Hence  overall communication would be $\mathcal{O}(d \cdot n^\frac{D-1}{D})$. This implies that when $d = o(n^\frac{1}{D})$, then communication is $o(n)$. The rest of the argument follows similar to the soundness proof of Theorem~\ref{thm: main}.
\end{proof}

\subsection{Wider class of Bell games} \label{subsec: wide_Bell games}
Our results can be applied to wider class of two-prover Bell games $\mathcal{G}=\left(V,\pi\right)$ where $V$ is a  predicate $V:\mathcal{X}\times\mathcal{Y\times\mathcal{A}}\times\mathcal{B}\rightarrow\{0,1\}$
and $\pi $ is a probability distribution on $\mathcal{X}\times\mathcal{Y}$. To extend our result to wider class of Bell games, we use the following Fact.

\begin{fact}[\cite{jain2022direct}] \label{fact:jkleakageparallel_general}
    The probability of classical players winning $(\omega_c(\mathcal{G}) + \eta)n$ games out of $n$ parallel copies of a two prover Bell game $\mathcal{G}=\left(V,\pi\right)$ with predicate $V$ and $\pi$ being a distribution on $\mathcal{X}\times\mathcal{Y}$ with classical value $\omega_c(\mathcal{G})$, with $c n$ bits being interactively leaked is $2^{-\Omega(\eta^3 n)+cn}$.
\end{fact}

The relation problem PARMAGIC can be modified as follows.

\begin{definition} [PARBELL $(\delta)$] Given line input $x, y\in  \{0,1\}^{2n}$, the goal is to obtain $a, b\in  \{0,1\}^{2n}$ such that $n(\omega_q(\mathcal{G})-\delta)$ tuples $(x_i,y_i,a_i,b_i)$ satisfy the Bell game $\mathcal{G}=\left(V,\pi\right)$ predicate $V(x_i,y_i,a_i,b_i)$, where $\pi$ is a distribution on $\mathcal{X}\times\mathcal{Y}$ and $\omega_q(\mathcal{G})$  is the quantum value.
\end{definition}
Moreover, we will need the following assumption.
\begin{assume} \label{assume: circuit_depth}
    The measurement settings that achieve the quantum value $\omega_q(\mathcal{G})$ for the Bell game corresponding to $\mathrm{PARBELL(\delta)}$ can be implemented using $1D$, geometrically-local, constant depth,  uniform quantum circuits with  fan-in $4$ gates.
\end{assume}
Armed with Assumption~\ref{assume: circuit_depth}, we get the following extension of Theorem~\ref{thm: main}.

\begin{theorem} \label{thm: wider_Bell}
Let $\delta \in [0, \omega_q(\mathcal{G})-\omega_c(\mathcal{G}))$ be a constant. For the $\mathrm{PARBELL(\delta)}$ problem corresponding to line-input of size $n$, and satisfying Assumption~\ref{assume: circuit_depth}, the following holds: 
\begin{enumerate}
    \item (Completeness) $\mathrm{PARBELL(\delta)} \in \mathrm{\{gl\}FQU}\left(O(1), \delta/100, 1-2^{-\Omega(n)} \right)$, 
    
    where $ \mathrm{\{gl\}FQU}\left(d, \eps, p \right)$ represents the class of relation problems that can be solved by $1D$, geometrically-local, depth $d$,  uniform quantum circuits with  fan-in $4$ gates, per gate noise  $\eps$ and probability of success at least $p$ (worst-case input).
    \item (Soundness) $\mathrm{PARBELL(\delta)} \notin \mathrm{\{gl\}F}\left(o(n), 2^{-o(n)} \right)$, 
    
    where 
    $\mathrm{\{gl\}F}\left(d, p \right)$ represents the class of relation problems that can be solved by $1D$,  geometrically-local, depth $d$, noiseless, non-uniform, classical circuits with fan-in $2$ NAND gates, with success probability at least $p$ (averaged over the inputs drawn uniformly). 
\end{enumerate}
\end{theorem}

\begin{proof}
The proof is similar to the proof of Theorem~\ref{thm: main} and follows from Fact~\ref{fact:jkleakageparallel_general} and Assumption~\ref{assume: circuit_depth}.
\end{proof}

\section{Open problems} \label{sec: discuss}
Our work leads to several natural open problems, mentioned below.

\paragraph{Using quantum contextuality:}
In future, it would be interesting to explore the wider ramifications of geometric locality in complexity theory. Further, in this work, we used Bell nonlocality as one of the important ingredients to obtain our central result. Bell nonlocality is a special case of quantum contextuality. In Ref~\cite{aasnaess2022comparing}, the author notes that quantum circuit and computational problem discussed in Ref.~\cite{bravyi2018quantum}  is derived from a
family of non-local games related to the well known GHZ non-local game, which can be thought of as special instances of quantum contextuality.
They present a generalised version of the construction in Ref.~\cite{bravyi2018quantum} as well as  a systematic way
of taking examples of contextuality and producing unconditional quantum
advantage results with shallow circuits. It would be of interest to investigate the potential applicability of concepts from Ref~\cite{aasnaess2022comparing} to improve the results presented in our paper by using tools from quantum contextuality.

\paragraph{Light cone arguments:}
The results in Ref.~\cite{bravyi2018quantum,coudron2021trading,bravyi2020quantum}
rely on lightcone arguments. Given an input bit $x_j$, the forward lightcone $L^{\rightarrow}_C (x_j)$ is the set of output bits $z_k$ such that ${x_j,z_k}$ are correlated. Similarly, for an output bit $z_k$, the backward lightcone $L^{\leftarrow}_C (z_k)$ is the set of input bits $x_j$ such that ${x_j,z_k}$ are correlated. It is not obvious how to concurrently apply parallel repetition and lightcone arguments, when starting with shared randomness, to amplify the soundness guarantees. It would be interesting to explore if ideas based on lightcone arguments can help recover or improve the results in this work.

\paragraph{Enforcing geometric locality:}
Enforcing geometric locality on prover's side is challenging and hence it remains open how one can design quantum advantage protocols that benefit from geometric locality.

\section*{Acknowledgement}
K.B. expresses gratitude to Dax Koh for engaging in discussions regarding circuit complexity and also acknowledges Atul Singh Arora for their insightful discussions on Bell games.

The work of R.J. is supported by the NRF grant NRF2021-QEP2-02-P05 and the Ministry
of Education, Singapore, under the Research Centres of Excellence program. This work was
done in part while R.J. was visiting the Technion-Israel Institute of Technology, Haifa, Israel and the Simons Institute for the Theory of Computing, Berkeley, CA, USA.

\bibliographystyle{alphaurl}
\bibliography{circuit_lower}

\end{document}